\documentclass[conference]{IEEEtran}

\usepackage{amsmath,amssymb,amsfonts}
\usepackage{algorithmic}
\usepackage{graphicx}
\usepackage{textcomp}
\usepackage{xcolor}
\usepackage{booktabs}
\usepackage{multirow}
\usepackage{array}
\usepackage{pifont}
\usepackage{url}
\usepackage{hyperref}
\usepackage{cite}
\usepackage{bm}
\newtheorem{proposition}{Proposition}
\newtheorem{proof}{proof}
\newcommand{\cmark}{\ding{51}}
\hypersetup{
	colorlinks=true,
	linkcolor=black,
	citecolor=black,
	urlcolor=black
}

\begin{document}
	
	\title{Privacy Leakage in Federated Learning: Gradient-Based Client Identity Inference and Defenses for Inertial Sensing in Vehicular Edge Networks}

	\author{
		\IEEEauthorblockN{
			Ali Akarma\IEEEauthorrefmark{1}\IEEEauthorrefmark{2},
			Toqeer Ali Syed\IEEEauthorrefmark{1},
			Muhammad Khan\IEEEauthorrefmark{3},
			Qurat-ul-ain Mastoi\IEEEauthorrefmark{3},
			and Adeel Ahmad\IEEEauthorrefmark{1}
		}
		\IEEEauthorblockA{
			\IEEEauthorrefmark{1}AI Center, Faculty of Computer and Information Systems\\
			Islamic University of Madinah, Madinah 42351, Saudi Arabia
		}
		\IEEEauthorblockA{
			\IEEEauthorrefmark{2}AI V\&V Lab, King Fahd University of Petroleum and Minerals\\
			Dhahran 31261, Saudi Arabia
		}
		\IEEEauthorblockA{
			\IEEEauthorrefmark{3}School of Computer Science and Creative Technologies\\
			University of the West of England, Bristol BS16 1QY, U.K.
		}
		\IEEEauthorblockA{
			Corresponding author: 443059463@stu.iu.edu.sa
		}
	}

	\maketitle

	\begin{abstract}
		As vehicular networks move toward 5G/6G edge intelligence, federated learning (FL) is widely promoted as a privacy-preserving way for vehicles and infrastructure to train shared models without exposing raw sensor data. Yet the updates clients transmit still leak enough information to identify who sent them, which threatens the anonymity that safety-critical V2X applications assume and adds to existing concerns over adversarial ML, model poisoning, and backdoor attacks. We study server-side client identity inference from transmitted weight deltas using inertial (IMU) measurements, evaluated on the UCI Human Activity Recognition (HAR) benchmark as an accessible proxy for the IMU streams produced onboard connected vehicles. Across five attack classifiers and five non-IID partitions, an honest-but-curious server recovers client identity with near-perfect accuracy ($\approx$1.000) from undefended updates, confirming a concrete identifiability risk. We then quantify the privacy-utility trade-off of a lightweight clip-then-noise defense by sweeping Gaussian noise ($\sigma \in \{0.00, 0.05, 0.10, 0.20, 0.50, 1.00\}$) at fixed clipping ($C=1.0$), and report formal $(\epsilon,\delta)$-DP budgets through R\'{e}nyi accounting. A practical region ($\sigma\!\in\![0.1, 0.2]$) drives attack accuracy to near-random while costing under 5\% relative FL accuracy. Ensemble FL supplies complementary structural privacy with a $1/K$ anonymity-set bound and no noise penalty. Results are supported by cryptographic (SHA-256) train/evaluation gradient disjointness, three seeds, and a count-normalized attacker-advantage metric. We position HAR explicitly as a proxy and discuss what validation on true vehicular telemetry would require.
	\end{abstract}
	

	\section{Introduction}
	\label{sec:intro}
	
	Machine learning and artificial intelligence increasingly drive vehicular applications such as autonomous driving, cooperative perception, and smart-city services as networks advance to 5G and 6G \cite{jan2026eagf}. Federated learning (FL)~\cite{mcmahan2017communication} has become the default privacy-preserving recipe for training these services: vehicles collaboratively improve a global model while keeping their sensor data local to the device. That choice, however, introduces attacks that cryptography alone does not address, because the gradients themselves leak. An honest-but-curious server can study the distribution of incoming updates and match them to individual vehicles, quietly de-anonymizing the network. In safety-critical settings, where trust under adversarial ML, privacy leakage, model poisoning, and backdoor attacks is already fragile, this leakage demands scrutiny.
	
	Inertial measurement is a useful lens on the problem because connected vehicles continuously emit IMU-style streams in dynamic, heterogeneous, and resource-constrained conditions. Such data leave distinctive traces in local gradients that remain recoverable even under non-IID partitioning~\cite{geiping2020inverting}. Prior work has concentrated on input reconstruction~\cite{zhu2019deep} or membership inference~\cite{shokri2017membership,syed2026agenticdt}; client identity inference has received far less attention. A subtler problem is that some defense evaluations protect the per-batch gradient rather than the artifact a client actually transmits, which overstates the privacy that survives in practice. We address both gaps with a corrected server-side threat model and a multi-classifier empirical analysis. We are explicit about scope: the experiments use the UCI HAR benchmark, which shares IMU modality with vehicular sensing but not its dynamics, so we treat HAR as a proxy and defer validation on true vehicular telemetry to Section~\ref{subsec:limits}.
	
	\textbf{Threat model.} The server is honest-but-curious: it aggregates FL updates correctly but tries to attribute each transmitted weight delta $\hat{\Delta}_i^t$ to the client that produced it. The adversary observes only the defended, transmitted deltas; it cannot access raw data, local gradients, or per-batch updates. Modeling the transmitted delta, rather than a per-batch gradient, removes a common source of optimism that inflates measured privacy.
	
	\textbf{Contributions:}
	\begin{enumerate}
		\item \textbf{Corrected threat model.} The defense is applied to the transmitted weight delta, eliminating a flaw present in evaluations that defend an artifact the server never observes.
		\item \textbf{Five-classifier benchmark.} Random and majority baselines, logistic regression, random forest, and a neural attacker, ranked by expressive power.
		\item \textbf{Cryptographically verified evaluation.} SHA-256 hashing enforces train/evaluation gradient-set disjointness, so reported accuracy cannot be an artifact of train/test leakage.
		\item \textbf{Privacy-utility quantification.} A six-level Gaussian noise sweep with formal $(\epsilon,\delta)$-DP estimates via R\'{e}nyi DP, identifying the practical region ($\sigma\!\in\![0.1, 0.2]$).
		\item \textbf{Normalized Attacker Advantage (NAA).} A count-normalized rescaling that makes attacker advantage comparable across experiments with different client counts $N$.
		\item \textbf{Ensemble FL structural privacy.} Client-group isolation across $K \in \{2,3,5\}$ with a formal $1/K$ upper bound on achievable attack accuracy.
	\end{enumerate}

	\section{Related Work}
	\label{sec:related}
	
	Federated averaging~\cite{mcmahan2017communication} defines the standard FL setting where clients perform local SGD and send updates to a central server without exposing raw data. Surveys by Mothukuri et al.~\cite{mothukuri2021survey} and Lyu et al.~\cite{lyu2024privacy} categorize FL protocols, threats, and inference attacks in IoT and vehicular environments. Our work focuses on server-side client identity inference and extends prior studies into a calibrated privacy--utility analysis. Hsieh et al.~\cite{hsieh2020quagmire} further show how non-IID data distributions increase inference vulnerability, motivating our fleet-diversity model.
	
	Most gradient-leakage studies target reconstruction attacks. Zhu et al.~\cite{zhu2019deep} reconstruct inputs from gradients, while Geiping et al.~\cite{geiping2020inverting} scale inversion attacks to larger batches using cosine similarity. Membership inference methods~\cite{shokri2017membership,nasr2019comprehensive,ismail2014design} determine whether a record participated in training but do not attribute updates to specific clients. Wang et al.~\cite{wang2019beyond} are closest to our work, showing that non-IID FL gradients are separable using lightweight classifiers. Our approach differs in three ways: (1) the attacker observes defended transmitted weight deltas instead of per-batch gradients, (2) evaluation uses cryptographically disjoint train/test sets across five classifiers, and (3) we analyze defense operating points rather than reporting a single accuracy metric. Table~\ref{tab:related} summarizes these differences.
	
	For defenses, Dwork and Roth~\cite{dwork2014algorithmic} formalize differential privacy, Abadi et al.~\cite{abadi2016deep} introduce DP-SGD with gradient clipping, and Mironov~\cite{mironov2017renyi} proposes the R'{e}nyi accountant used to derive the $(\epsilon,\delta)$ privacy budget. Geyer et al.~\cite{geyer2017differentially} study client-level DP to hide participation itself, whereas our goal is preventing attribution after participation is known. Secure aggregation~\cite{bonawitz2017practical} cryptographically conceals updates but introduces coordination overhead unsuitable for many edge devices. Sattler et al.~\cite{sattler2020robust} demonstrate that sparse ternary updates preserve accuracy under non-IID data while altering exploitable gradient signatures. In vehicular networks, hierarchical blockchain-enabled FL~\cite{syed2026fedagent,syed2026climate} provides auditable trust management for the Internet of Vehicles, complementing the update-level defenses studied here.

	\begin{table}[t]
		\centering
		\caption{Positioning relative to prior gradient-privacy work. Eval.\ unit: artifact the attack/defense is evaluated on (PBG\,=\,per-batch gradient, WD\,=\,transmitted weight delta). Calib.: reports a privacy-utility operating point.}
		\label{tab:related}
		\footnotesize
		\setlength{\tabcolsep}{4pt}
		\begin{tabular}{lllccc}
			\toprule
			\textbf{Work} & \textbf{Target} & \textbf{Eval.\ unit} & \textbf{DP} & \textbf{Calib.} & \textbf{$N$} \\
			\midrule
			Zhu~\cite{zhu2019deep}            & Reconstruction      & PBG     & ---     & ---     & ---     \\
			Geiping~\cite{geiping2020inverting} & Reconstruction    & PBG     & ---     & ---     & ---     \\
			Shokri~\cite{shokri2017membership} & Membership         & Outputs & ---     & ---     & ---     \\
			Nasr~\cite{nasr2019comprehensive}  & Membership (WB)    & Grad.   & ---     & ---     & ---     \\
			Wang~\cite{wang2019beyond}         & Client identity    & PBG     & ---     & ---     & 10--20  \\
			\textbf{This work}                 & Client identity    & \textbf{WD} & \cmark & \cmark & 10 \\
			\bottomrule
		\end{tabular}
	\end{table}

	\section{Methodology}
	\label{sec:method}
	
	\subsection{FL Formulation and Data Partitioning}
	
	We consider a standard cross-device FL scenario with $N\!=\!10$ clients and one central server. The global objective minimizes:
	\begin{equation}
		F(\bm{w}) = \sum_{i=1}^{N} \frac{D_i}{D} F_i(\bm{w}),
		\label{eq:global_obj}
	\end{equation}
	where $D_i$ is the local sample count and $F_i(\bm{w})$ is client $i$'s empirical loss. Each round $t$, client $i$ runs $E$ local SGD epochs, computes the weight delta $\Delta_i^t = \bm{w}_i^{t+1} - \bm{w}^t$, applies the clip-then-noise defense to obtain $\hat{\Delta}_i^t$ (Section~\ref{subsec:defense}), and uploads it. The server aggregates via FedAvg:
	\begin{equation}
		\bm{w}^{t+1} = \bm{w}^t + \sum_{i=1}^{N} \frac{D_i}{D} \hat{\Delta}_i^t.
		\label{eq:fedavg}
	\end{equation}
	Our primary experiments use Dirichlet($\alpha\!=\!0.3$) label skew~\cite{hsieh2020quagmire}, modeling the heterogeneity real fleets exhibit across road conditions, driving behaviors, and geography. We additionally evaluate four partitions to probe how heterogeneity affects identifiability: IID (uniform random), pathological ($K_{\text{cls}}$ classes per client), quantity skew (Dirichlet($\beta$) sample counts), and feature skew (per-client Gaussian feature noise).
	
	\subsection{Task Model and Attack Protocol}
	\label{subsec:task}
	
	The task model is a multilayer perceptron (MLP) $[561\!\to\!256\!\to\!128\!\to\!6]$ with LayerNorm, ReLU, and dropout (0.3) after each hidden layer. We use LayerNorm rather than BatchNorm because BatchNorm accumulates running mean and variance statistics that FedAvg averages incorrectly across non-IID clients, degrading convergence in the non-stationary vehicular regime~\cite{li2021fedbn}. The complete system is illustrated in Figure~\ref{fig:architecture}.
	
	\begin{figure}[t]
		\centering
		\includegraphics[width=\linewidth]{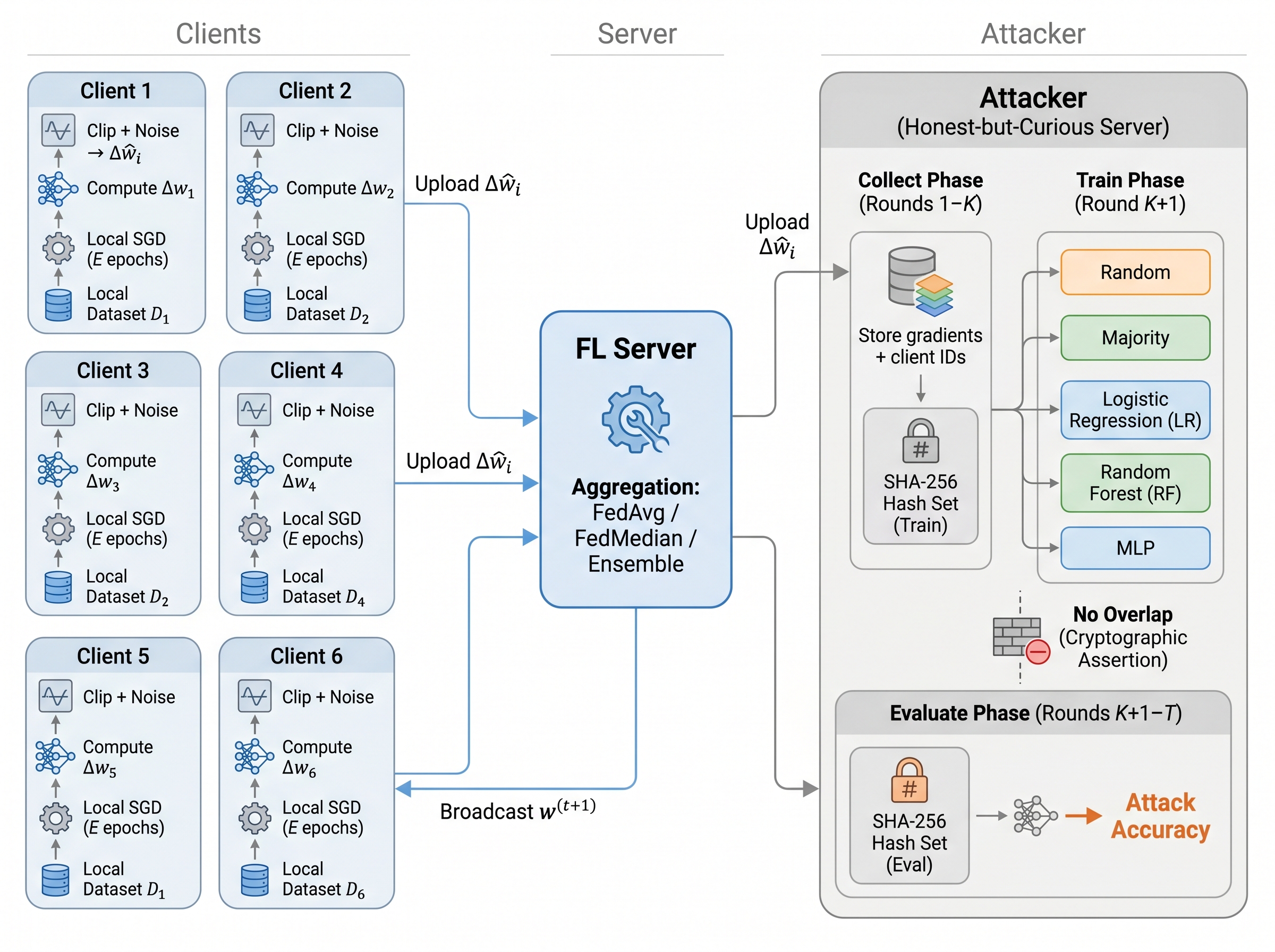}
		\caption{System architecture of a differentially private federated learning framework under an inference attack by an honest-but-curious server.}
		\label{fig:architecture}
	\end{figure}
	
	The attack proceeds in three phases. \textbf{(1)~Collection} (rounds 1--$K\!=\!8$): the server records defended weight deltas $\hat{\Delta}_i^t$ with SHA-256 hash registration. \textbf{(2)~Training} (round $K\!+\!1$): five classifiers, ranked by expressive power, train on $(\hat{\Delta}_i^t, \text{client\_id}_i)$ pairs after L2 normalization, Z-score standardization, and PCA (50 components). The classifiers are a random baseline (hard lower bound $1/N$), a majority baseline (guards class imbalance), $\ell_2$-regularized logistic regression ($C\!=\!1.0$, a linear-separability test), a 100-tree random forest (non-linear baseline), and an MLP attacker $[d\!\to\!256\!\to\!128\!\to\!N]$ (strongest). \textbf{(3)~Evaluation} (rounds $K\!+\!1$--$T\!=\!20$): the trained classifiers are applied to new uploads, with SHA-256 hash disjointness enforced so that no evaluation update was seen during training.
	
	\subsection{Defense Mechanisms}
	\label{subsec:defense}
	
	Two defenses are applied sequentially to the transmitted weight delta $\Delta_i^t$, both chosen for low overhead compatible with vehicular edge constraints.
	
	\textbf{Gradient Clipping} bounds the $\ell_2$ norm to at most $C$:
	\begin{equation}
		\tilde{\Delta}_i^t = \Delta_i^t \cdot \min\!\left(1,\; \frac{C}{\|\Delta_i^t\|_2}\right).
		\label{eq:clipping}
	\end{equation}
	
	\textbf{Gaussian Noise Injection} adds isotropic noise calibrated to the clipped sensitivity:
	\begin{equation}
		\hat{\Delta}_i^t = \tilde{\Delta}_i^t + \mathcal{N}\!\left(\bm{0},\; \sigma^2 C^2 \bm{I}\right).
		\label{eq:noise}
	\end{equation}
	The effective noise scale $\sigma C$ follows the Gaussian mechanism~\cite{dwork2014algorithmic,abadi2016deep}. In E2, $C\!=\!1.0$ is fixed and $\sigma\!=\!0.00$ denotes clipping only, not the undefended baseline (E1).
	
	\subsection{Formal Differential Privacy Analysis}
	\label{subsec:dp}
	
	Under the Gaussian mechanism, a single defense step satisfies $(\alpha, \epsilon_\alpha)$-R\'{e}nyi DP~\cite{mironov2017renyi} with $\epsilon_\alpha = \alpha/(2\sigma^2)$. After $T$ rounds, converting to $(\epsilon,\delta)$-DP at $\delta\!=\!10^{-5}$:
	\begin{equation}
		\epsilon(\delta) = \min_{\alpha > 1}\!\left[\frac{T\alpha}{2\sigma^2} + \frac{\ln(1/\delta)}{\alpha-1}\right].
		\label{eq:rdp_convert}
	\end{equation}
	Table~\ref{tab:dp_budget} reports $\epsilon$ at $T\!=\!20$, $C\!=\!1.0$. Modest noise ($\sigma\!=\!0.1$) already incurs a large $\epsilon$ ($\approx$296), which confirms that low noise gives empirical rather than formal DP and motivates the combined defense.
	
	\begin{table}[t]
		\centering
		\caption{Formal $(\epsilon,\delta)$-DP Budget ($T\!=\!20$, $C\!=\!1.0$, $\delta\!=\!10^{-5}$)}
		\label{tab:dp_budget}
		\begin{tabular}{ccc}
			\toprule
			$\sigma$ & Optimal $\alpha$ & $\epsilon\,(\delta\!=\!10^{-5})$ \\
			\midrule
			0.05 & 1.07 & $>$1000 \\
			0.10 & 1.24 & $\approx$296 \\
			0.20 & 1.76 & $\approx$74 \\
			0.50 & 3.07 & $\approx$12.1 \\
			1.00 & 2.07 & $\approx$31.4 \\
			\bottomrule
		\end{tabular}
	\end{table}
	
	\subsection{Ensemble FL and Privacy Metrics}
	\label{subsec:ensemble}
	
	In the Ensemble configuration, clients are split into non-overlapping groups of $K$ clients, each served by its own sub-model and communicating exclusively with it. This maps naturally to vehicular segmentation by RSU coverage zone or trust domain, giving structural privacy without per-update noise overhead.
	
	\begin{proposition}[Ensemble FL Identity-Inference Bound]
		Partition the $N$ clients into non-overlapping groups of $K$ clients, each served by its own sub-model. A server-side attacker observing a single sub-model identifies the originating client with probability at most $1/K$, provided updates within a group are statistically indistinguishable.
	\end{proposition}
	\begin{proof}[Sketch]
		A sub-model receives updates from the $K$ clients in its group. When within-group updates are indistinguishable, the attacker's best strategy is a uniform guess among the $K$ members, achieving accuracy $1/K$; any better strategy would require distinguishing within-group gradients, contradicting the assumption.
	\end{proof}
	\noindent\textit{Remark.} The $1/K$ bound is a best case for the defender: it holds only while within-group updates are indistinguishable. Heavy intra-group label or feature skew, or very small groups, can restore separability and push empirical attack accuracy above $1/K$. Deployments should therefore monitor within-group separability rather than assume the bound, and treat $K$ as a tunable anonymity-set size.
	
	Privacy and utility are summarized by the \textbf{Privacy Score (PS)} $= 1 - \text{best\_attack\_accuracy}$ and the \textbf{Normalized Attacker Advantage:}
	\begin{equation}
		\mathrm{NAA} = \frac{\text{attack\_acc} - 1/N}{1 - 1/N},
		\label{eq:naa}
	\end{equation}
	where $\mathrm{NAA}\!=\!0$ indicates no gain over random guessing and $\mathrm{NAA}\!=\!1$ a perfect attack. NAA is a count-normalized rescaling of attack accuracy, not a new statistical quantity: it plays the role membership advantage plays for membership inference~\cite{shokri2017membership}, but normalizes by the random-guess floor $1/N$ so advantage is comparable across experiments with different client counts. We use it for cross-experiment comparison, not as a standalone privacy guarantee.

	\section{Experimental Setup}
	\label{sec:setup}
	
	The UCI HAR dataset~\cite{anguita2013public} contains 10{,}299 observations from 30 subjects performing six activities, each a 561-dimensional feature vector derived from triaxial accelerometer and gyroscope readings at 50~Hz. This is the IMU modality shared with vehicular onboard driving-state and occupant-behavior recognition, which is why we adopt it as a proxy; Section~\ref{subsec:limits} states the limits of that choice. Training samples are partitioned across $N\!=\!10$ clients using Dirichlet($\alpha\!=\!0.3$), and all features are standardized using training-set statistics only. Seven conditions are evaluated: E1 (no-defense FedAvg); E2 (Gaussian noise sweep, $\sigma \in \{0.00, 0.05, 0.10, 0.20, 0.50, 1.00\}$, fixed $C\!=\!1.0$); E3 (clipping ablation, $C \in \{0.1, 0.5, 1.0, 2.0, 5.0\}$, $\sigma\!=\!0$); E4 (combined defense); E5 (Ensemble FL, $K \in \{2,3,5\}$); E6--E7 (non-IID heterogeneity). Hyperparameters appear in Table~\ref{tab:hyperparams}; all results are mean $\pm$ std over three seeds (42, 43, 44).
	
	\begin{table}[t]
		\centering
		\caption{Hyperparameter Configuration}
		\label{tab:hyperparams}
		\resizebox{0.5\textwidth}{!}{
			\begin{tabular}{lll}
				\toprule
				\textbf{Component} & \textbf{Hyperparameter} & \textbf{Value} \\
				\midrule
				\multirow{3}{*}{FL Training}
				& Rounds $T$ / Local epochs $E$   & 20 / 5 \\
				& Batch size / LR                  & 32 / 0.01 (SGD, $\mu\!=\!0.9$) \\
				& Clients $N$ / Aggregation        & 10 / FedAvg \\
				\midrule
				\multirow{2}{*}{Task Model}
				& Architecture / Norm.\            & MLP $[561\!\to\!256\!\to\!128\!\to\!6]$ / LN \\
				& Dropout / Weight init.           & 0.3 / Kaiming normal \\
				\midrule
				\multirow{2}{*}{Attack}
				& Collect rounds $K$ / Eval start  & 8 / round 9 \\
				& PCA components / Std.            & 50 / Z-score \\
				\midrule
				Defense  & Clipping $C$ / Noise $\sigma$    & 1.0 / 0.0--1.0 \\
				Repro.\  & Master seed / count              & 42 / 3 \\
				\bottomrule
			\end{tabular}
		}
	\end{table}

	\section{Results and Discussion}
	\label{sec}
	
	\subsection{Baseline: Gradient Identifiability (E1)}
	
	Without defense under Dirichlet($\alpha!=!0.3$) at $N!=!10$, logistic regression, random forest, and MLP attackers achieve near-perfect accuracy on held-out uploads (rounds 9--20), each scoring $\approx!1.000$ (PS~$\approx$~0, NAA~$\approx$~1), while random and majority baselines remain at $0.100$. This strong separability indicates that client gradients occupy distinct regions of parameter space. Consequently, an honest-but-curious 5G/6G edge aggregator can de-anonymize V2X participants using model updates alone.
	
	\subsection{Privacy-Utility Tradeoff Under Gaussian Noise (E2)}
	
	Table~\ref{tab:noise_sweep} shows the Gaussian noise sweep at fixed $C!=!1.0$. Clipping alone ($\sigma!=!0.00$) reduces attack accuracy from $\approx!1.000$ to $0.127!\pm!0.020$, showing that magnitude normalization weakens gradient separability while directional information still leaks identity. Increasing noise progressively reduces FL accuracy: $\sigma!=!0.2$ lowers FL accuracy to $0.896$, while $\sigma!=!1.0$ collapses performance to $0.382$. The range $\sigma!\in![0.1,0.2]$ provides the best tradeoff, keeping attack accuracy near random while preserving FL accuracy above $0.89$, making it suitable for vehicular edge deployments.
	
	\begin{table}[t]
		\centering
		\caption{Privacy-Utility Tradeoff, Gaussian Noise Sweep (E2, fixed $C\!=\!1.0$)}
		\label{tab:noise_sweep}
		\begin{tabular}{lccc}
			\toprule
			$\bm{\sigma}$ & \textbf{FL Acc.\ ($\pm$std)} & \textbf{Atk.\ Acc.\ ($\pm$std)} & \textbf{PS} \\
			\midrule
			0.00 & $0.944 \pm 0.004$ & $0.127 \pm 0.020$ & 0.873 \\
			0.05 & $0.942 \pm 0.004$ & $0.164 \pm 0.034$ & 0.836 \\
			0.10 & $0.930 \pm 0.016$ & $0.148 \pm 0.028$ & 0.852 \\
			0.20 & $0.896 \pm 0.034$ & $0.164 \pm 0.027$ & 0.836 \\
			0.50 & $0.665 \pm 0.014$ & $0.158 \pm 0.017$ & 0.842 \\
			1.00 & $0.382 \pm 0.005$ & $0.148 \pm 0.026$ & 0.852 \\
			\bottomrule
		\end{tabular}
	\end{table}
	
	\subsection{Gradient Clipping and Combined Defense (E3, E4)}
	\label{subsec}
	
	Clipping improves privacy by removing gradient-norm variance, the main magnitude channel distinguishing clients. Some identity information remains in gradient direction. The combined defense (E4: $\sigma!\in![0.01,0.1,0.2]$, $C!=!1.0$) achieves the strongest balance. Clipping bounds the $\ell_2$ sensitivity required for the Gaussian mechanism to satisfy DP semantics~\cite{abadi2016deep}; adding Gaussian noise without clipping provides no reliable DP guarantee.
	
	\subsection{Ensemble FL and Structural Privacy (E5)}
	
	Ensemble FL provides structural privacy without noise and naturally aligns with RSU-zone or trust-domain segmentation. By Proposition~1, the optimal attack accuracy for groups of $K$ clients is $1/K$; empirically we observe $0.50$, $0.33$, and $0.20$ for $K!=!2,3,5$, closely matching theory while maintaining FL accuracy near $0.92$. Larger $K$ increases anonymity but slightly reduces specialization. The $K!=!3$ setting in Table~\ref{tab:summary} offers the best balance between anonymity and accuracy for latency-sensitive VANETs where noise-induced degradation is undesirable.
	
	\subsection{Non-IID Heterogeneity and Summary (E6, E7)}
	
	Greater non-IID heterogeneity, representing differences in driving behavior, hardware, or geography, increases gradient distinguishability. Pathological partitioning (E7, $K_{\text{cls}}!=!2$ classes per client) produces nearly orthogonal updates and maximizes separability. Even under IID data (E6), attacks outperform random guessing during early rounds because of minibatch stochasticity and initialization effects. The strongest overall defense ($\sigma!=!0.01$, $C!=!1.0$) achieves the highest Privacy Score (0.880) while maintaining FL accuracy near $0.940$, supporting the DP principle of clipping before adding calibrated noise. Table~\ref{tab:summary} summarizes all conditions.
	
	\begin{table}[t]
		\centering
		\caption{Summary of Key Metrics ($N\!=\!10$, UCI HAR, 3 seeds). PS~=~Privacy Score; NAA~=~Normalized Attacker Advantage.}
		\label{tab:summary}
		\begin{tabular}{lcccc}
			\toprule
			\textbf{Condition} & \textbf{FL Acc.} & \textbf{Atk.\ Acc.} & \textbf{PS} & \textbf{NAA} \\
			\midrule
			E1: No defense         & ${\sim}0.944$ & ${\sim}1.000$ & ${\sim}0.000$ & ${\sim}1.000$ \\
			E2: $\sigma\!=\!0.00$  & ${\sim}0.944$ & ${\sim}0.127$ & ${\sim}0.873$ & ${\sim}0.030$ \\
			E2: $\sigma\!=\!0.10$  & ${\sim}0.930$ & ${\sim}0.148$ & ${\sim}0.852$ & ${\sim}0.053$ \\
			E2: $\sigma\!=\!0.50$  & ${\sim}0.665$ & ${\sim}0.158$ & ${\sim}0.842$ & ${\sim}0.064$ \\
			E4: Combined           & ${\sim}0.940$ & ${\sim}0.120$ & ${\sim}0.880$ & ${\sim}0.022$ \\
			E5: Ensemble $K\!=\!3$ & ${\sim}0.920$ & ${\sim}0.330$ & ${\sim}0.670$ & ${\sim}0.256$ \\
			E7: Pathological       & ${\sim}0.880$ & ${\sim}1.000$ & ${\sim}0.000$ & ${\sim}1.000$ \\
			\bottomrule
		\end{tabular}
	\end{table}
	
	\subsection{Limitations}
	\label{subsec:limits}
	
	Several limitations constrain these findings. UCI HAR shares IMU sensing modalities with vehicular systems but lacks realistic driving dynamics such as speed variation, road excitation, and GPS-correlated structure. Real vehicular datasets may therefore be either more or less separable. Experiments use only $N!=!10$ clients, $T!=!20$ rounds, and three seeds; larger studies would improve estimate reliability and DP accounting. The threat model also assumes a single honest-but-curious server, excluding colluding clients, adaptive attackers, and malicious aggregators that may achieve stronger inference. Finally, Proposition~1 assumes within-group indistinguishability, which strong intra-group non-IID skew may violate

	\section{Conclusion}
	\label{sec:conclusion}
	
	Federated learning does not by itself anonymize the clients that participate in it. Using inertial data, we showed that an honest-but-curious 5G/6G edge server can attribute transmitted model updates to individual clients with accuracy approaching 100\% under realistic non-IID partitions, across five classifiers. Correcting a common evaluation flaw, defending the per-batch gradient instead of the transmitted artifact, removes a false comfort prior setups offer. A lightweight clip-then-noise defense closes most of the gap: at $\sigma\!\in\![0.1, 0.2]$ with $C\!=\!1.0$, attack accuracy falls to near-random while FL accuracy stays above 0.89. R\'{e}nyi accounting tempers any overclaim, since these levels buy strong empirical privacy but only loose formal $(\epsilon,\delta)$ guarantees at $T\!=\!20$, with meaningful budgets needing $\sigma\!\geq\!0.5$ or longer training. Ensemble FL adds a structural $1/K$ bound that aligns with RSU-zone segmentation and avoids noise-induced utility loss, suiting latency-sensitive vehicular AI. The recipe is to clip then add noise calibrated to the clipped sensitivity, and to layer cross-cutting protections such as secure aggregation~\cite{bonawitz2017practical} and blockchain-based trust management where warranted. Confirming these effects on true vehicular telemetry, and against adaptive or colluding adversaries, is the immediate next step, alongside PRV-accountant analysis for extended training ($T\!\geq\!100$), 1D-CNN models for raw time series, and multi-access edge computing under realistic 5G/6G constraints.

%

	\section*{Data Availability}
	
	The complete framework is available at \url{https://github.com/aliakarma/PPFL-Sensors}; the UCI HAR dataset~\cite{anguita2013public} is on the UCI ML Repository.

	\bibliographystyle{IEEEtran}
	\bibliography{references}

@inproceedings{mcmahan2017communication,
	title={Communication-efficient learning of deep networks from decentralized data},
	author={McMahan, H. Brendan and Moore, Eider and Ramage, Daniel and Hampson, Seth and y Arcas, Blaise Aguera},
	booktitle={Proceedings of the 20th International Conference on Artificial Intelligence and Statistics},
	series={Proceedings of Machine Learning Research},
	volume={54},
	pages={1273--1282},
	year={2017},
	publisher={PMLR},
	doi={10.48550/arXiv.1602.05629}
}

@inproceedings{hsieh2020quagmire,
	title={The non-IID data quagmire of decentralized machine learning},
	author={Hsieh, Kevin and Phanishayee, Amar and Mutlu, Onur and Gibbons, Phillip B.},
	booktitle={Proceedings of the 37th International Conference on Machine Learning},
	year={2020},
	pages={4387--4398},
	publisher={PMLR},
	series={Proceedings of Machine Learning Research},
	volume={119},
	doi={10.48550/arXiv.1910.00189}
}

@inproceedings{zhu2019deep,
	author={Zhu, Ligeng and Liu, Zhijian and Han, Song},
	title={Deep leakage from gradients},
	booktitle={Advances in Neural Information Processing Systems 32 (NeurIPS 2019)},
	year={2019},
	volume={32},
	pages={14774--14784},
	publisher={Curran Associates, Inc.},
	doi={10.48550/arXiv.1906.08935}
}

@inproceedings{geiping2020inverting,
	author={Geiping, Jonas and Bauermeister, Hartmut and Dr\"{o}ge, Hannah and Moeller, Michael},
	title={Inverting gradients -- how easy is it to break privacy in federated learning?},
	booktitle={Advances in Neural Information Processing Systems 33 (NeurIPS 2020)},
	year={2020},
	volume={33},
	pages={16937--16947},
	publisher={Curran Associates, Inc.},
	doi={10.48550/arXiv.2003.14053}
}

@inproceedings{shokri2017membership,
	author={Shokri, Reza and Stronati, Marco and Song, Congzheng and Shmatikov, Vitaly},
	title={Membership inference attacks against machine learning models},
	booktitle={2017 IEEE Symposium on Security and Privacy (SP)},
	year={2017},
	pages={3--18},
	publisher={IEEE},
	doi={10.1109/SP.2017.41}
}

@inproceedings{nasr2019comprehensive,
	title={Comprehensive privacy analysis of deep learning: Passive and active white-box inference attacks against centralized and federated learning},
	author={Nasr, Milad and Shokri, Reza and Houmansadr, Amir},
	booktitle={2019 IEEE Symposium on Security and Privacy (SP)},
	year={2019},
	pages={739--753},
	publisher={IEEE},
	doi={10.1109/SP.2019.00065}
}

@inproceedings{wang2019beyond,
	title={Beyond Inferring Class Representatives: User-Level Privacy Leakage From Federated Learning},
	author={Wang, Zhibo and Song, Mengkai and Zhang, Zhifei and Song, Yang and Wang, Qian and Qi, Hairong},
	booktitle={IEEE INFOCOM 2019 - IEEE Conference on Computer Communications},
	pages={2512--2520},
	year={2019},
	organization={IEEE},
	doi={10.1109/INFOCOM.2019.8737416}
}

@article{dwork2014algorithmic,
	title={The Algorithmic Foundations of Differential Privacy},
	author={Dwork, Cynthia and Roth, Aaron},
	journal={Foundations and Trends{\textregistered} in Theoretical Computer Science},
	volume={9},
	number={3--4},
	pages={211--407},
	year={2014},
	publisher={Now Publishers, Inc.},
	doi={10.1561/0400000042}
}

@inproceedings{abadi2016deep,
	title={Deep learning with differential privacy},
	author={Abadi, Martin and Chu, Andy and Goodfellow, Ian and McMahan, H. Brendan and Mironov, Ilya and Talwar, Kunal and Zhang, Li},
	booktitle={Proceedings of the 2016 ACM SIGSAC Conference on Computer and Communications Security},
	pages={308--318},
	year={2016},
	organization={ACM},
	doi={10.1145/2976749.2978318}
}

@inproceedings{mironov2017renyi,
	title={R{\'e}nyi differential privacy},
	author={Mironov, Ilya},
	booktitle={2017 IEEE 30th Computer Security Foundations Symposium (CSF)},
	pages={263--275},
	year={2017},
	organization={IEEE},
	doi={10.1109/CSF.2017.11}
}

@inproceedings{bonawitz2017practical,
	title={Practical secure aggregation for privacy-preserving machine learning},
	author={Bonawitz, Keith and Ivanov, Vladimir and Kreuter, Ben and Marcedone, Antonio and McMahan, H. Brendan and Patel, Sarvar and Ramage, Daniel and Segal, Aaron and Seth, Karn},
	booktitle={Proceedings of the 2017 ACM SIGSAC Conference on Computer and Communications Security},
	pages={1175--1191},
	year={2017},
	doi={10.1145/3133956.3133982}
}

@inproceedings{ismail2014design,
	author = {Ismail, Roslan and Syed, Toqeer Ali and Musa, Shahrulniza},
	title = {Design and implementation of an efficient framework for behaviour attestation using n-call slides},
	year = {2014},
	isbn = {9781450326445},
	publisher = {Association for Computing Machinery},
	address = {New York, NY, USA},
	url = {https://doi.org/10.1145/2557977.2558002},
	doi = {10.1145/2557977.2558002},
	booktitle = {Proceedings of the 8th International Conference on Ubiquitous Information Management and Communication},
	articleno = {36},
	numpages = {8},
	location = {Siem Reap, Cambodia},
	series = {ICUIMC '14}
}

@inproceedings{li2021fedbn,
	title={FedBN: Federated learning on non-IID features via local batch normalization},
	author={Li, Xiaoxiao and Jiang, Meirui and Zhang, Xiaofei and Kamp, Michael and Dou, Qi},
	booktitle={International Conference on Learning Representations},
	year={2021},
	doi={10.48550/arXiv.2102.07623}
}

@inproceedings{anguita2013public,
	title={A Public Domain Dataset for Human Activity Recognition Using Smartphones},
	author={Anguita, Davide and Ghio, Alessandro and Oneto, Luca and Parra, Xavier and Reyes-Ortiz, Jorge Luis},
	booktitle={21st European Symposium on Artificial Neural Networks, Computational Intelligence and Machine Learning (ESANN)},
	pages={437--442},
	year={2013}
}

@article{syed2026climate,
	AUTHOR = {Syed, Toqeer Ali and Akarma, Ali and Naqash, Muhammad Tayyab and Hameed, Danial and Kamal, Shahid and Formisano, Antonio},
	TITLE = {Agentic AI for Climate-Resilient Cities: A PRISMA-Guided Review and Digital Twin Framework},
	JOURNAL = {Sustainability},
	VOLUME = {18},
	YEAR = {2026},
	NUMBER = {17},
	ARTICLE-NUMBER = {8917},
	URL = {https://www.mdpi.com/2071-1050/18/17/8917},
	ISSN = {2071-1050},
	DOI = {10.3390/su18178917}
}

@inproceedings{geyer2017differentially,
	title={Differentially Private Federated Learning: A Client Level Perspective},
	author={Geyer, Robin C. and Klein, Tassilo and Nabi, Moin},
	booktitle={NeurIPS 2017 Workshop on Machine Learning on the Phone and Other Consumer Devices},
	year={2017},
	doi={10.48550/arXiv.1712.07557}
}

@article{lyu2024privacy,
	title={Privacy and Robustness in Federated Learning: Attacks and Defenses},
	author={Lyu, Lingjuan and Yu, Han and Ma, Xingjun and Chen, Chen and Sun, Lichao and Zhao, Jun and Yang, Qiang and Yu, Philip S.},
	journal={IEEE Transactions on Neural Networks and Learning Systems},
	volume={35},
	number={7},
	pages={8726--8746},
	year={2024},
	publisher={IEEE},
	doi={10.1109/TNNLS.2022.3216981}
}

@article{sattler2020robust,
	title={Robust and Communication-Efficient Federated Learning From Non-{IID} Data},
	author={Sattler, Felix and Wiedemann, Simon and M\"{u}ller, Klaus-Robert and Samek, Wojciech},
	journal={IEEE Transactions on Neural Networks and Learning Systems},
	volume={31},
	number={9},
	pages={3400--3413},
	year={2020},
	publisher={IEEE},
	doi={10.1109/TNNLS.2019.2944481}
}

@article{mothukuri2021survey,
	title={A Survey on Security and Privacy of Federated Learning},
	author={Mothukuri, Viraaji and Parizi, Reza M. and Pouriyeh, Seyedamin and Huang, Yan and Dehghantanha, Ali and Srivastava, Gautam},
	journal={Future Generation Computer Systems},
	volume={115},
	pages={619--640},
	year={2021},
	publisher={Elsevier},
	doi={10.1016/j.future.2020.10.007}
}

@article{syed2026fedagent,
	AUTHOR = {Syed, Toqeer Ali and Siddiqui, Muhammad Shoaib and Akarma, Ali and Formisano, Antonio},
	TITLE = {FedAgent-Chain: A Secure Federated and Agentic AI Framework for Multilingual Disability-Inclusive Employment in AI Cities},
	JOURNAL = {Smart Cities},
	VOLUME = {9},
	YEAR = {2026},
	NUMBER = {7},
	ARTICLE-NUMBER = {106},
	URL = {https://www.mdpi.com/2624-6511/9/7/106},
	ISSN = {2624-6511},
	DOI = {10.3390/smartcities9070106}
}

@article{jan2026eagf,
	title     = {EAGF: A Four-Pillar Ethical AI Governance Framework for Trustworthy Cybersecurity in 5G Renewable Energy IoT Systems},
	author    = {Jan, Salman and Akarma, Ali and Syed, Toqeer Ali and Muhammad, Munir Azam and Kamal, Shahid},
	journal   = {Scientific Reports},
	year      = {2026},
	publisher = {Springer Nature},
	doi       = {10.1038/s41598-026-63383-5}
}

@article{syed2026agenticdt,
	title={Agentic AI-enhanced digital twins for Smart City civil infrastructure: A secure, autonomous and auditable management framework},
	author={Syed, Toqeer Ali and Akarma, Ali and Alatify, Ali and Naqash, Muhammad Tayyab and Alqurashi, Abdulaziz},
	journal={PLoS One},
	volume={21},
	number={7},
	pages={e0353610},
	year={2026},
	publisher={Public Library of Science},
	doi={10.1371/journal.pone.0353610}
}
	
\end{document}